\documentclass[11pt,a4paper]{article}

\usepackage{amsmath}
\usepackage[T1]{fontenc}
\usepackage[utf8]{inputenc}
\usepackage{amsfonts}
\usepackage{amsthm}
\usepackage{amstext}
\usepackage{amssymb}
\usepackage{color}
\usepackage{bbold}
\usepackage{hyperref}
\usepackage{url}

\newcommand{\diag}{\operatorname{diag}}
\newcommand{\rk}{\operatorname{rank}}

\newtheorem{theorem}{Theorem}
\newtheorem{example}[theorem]{Example}

\newtheorem{lemma}[theorem]{Lemma}
\newtheorem*{lemma*}{Lemma}
\newtheorem{proposition}[theorem]{Proposition}
\newtheorem{corollary}[theorem]{Corollary}
\newtheorem{remark}[theorem]{Remark}

\newtheorem*{open*}{Open~question}

\DeclareMathOperator{\Ima}{Im}

\begin{document}

\title{On the Uniqueness and Computation of Commuting Extensions}

\author{Pascal Koiran\footnote{Univ Lyon, EnsL, UCBL, CNRS,  LIP, F-69342, LYON Cedex 07, France. 
Email: {\tt pascal.koiran@ens-lyon.fr}.}}

\maketitle

\begin{abstract}
A tuple $(Z_1,\ldots,Z_p)$  of matrices of size $r$ is said to be a {\em commuting extension} of a tuple $(A_1,\ldots,A_p)$ of matrices of size $n <r$ if the $Z_i$ pairwise commute 
and each $A_i$ sits in the upper left corner of a block decomposition of $Z_i$.
This notion was discovered and rediscovered in several contexts including algebraic complexity theory (in Strassen's work on tensor rank), in 
numerical analysis for the construction of cubature formulas  and  in quantum mechanics for the study of computational methods and the study of the so-called "quantum Zeno dynamics."
 Commuting extensions have also attracted the attention of the linear algebra community.
 In this paper we present 3 types of results:
\begin{itemize}
\item[(i)] Theorems on the uniqueness of commuting extensions for three matrices or more.
\item[(ii)] Algorithms for the computation of commuting extensions of minimal size. These algorithms work under the same assumptions
as our uniqueness theorems. They are applicable up to $r=4n/3$, and are apparently the first provably efficient algorithms for this problem applicable beyond $r=n+1$.
\item[(iii)] A genericity theorem showing that our algorithms and uniqueness theorems can be applied to a wide range of
input matrices.
\end{itemize}
\end{abstract}

\newpage

\section{Introduction}

A tuple $(Z_1,\ldots,Z_p)$  of matrices in $M_r(K)$ is said to be a {\em commuting extension} of a tuple $(A_1,\ldots,A_p)$ of matrices in $M_n(K)$ if the $Z_i$ pairwise commute 
and each $A_i$ sits in the upper left corner of a block decomposition of  $Z_i$, i.e.,
\begin{equation} \label{eq:block}
Z_i=
\begin{pmatrix}
A_i & B_i \\
C_i & D_i
\end{pmatrix}
\end{equation}
for some matrices $B_i \in M_{n,r-n}(K)$, $C_i \in M_{r-n,n}(K)$ and $D_i \in M_{r-n}(K)$.
Here we denote by $M_{r,s}(K)$ the set of matrices with $r$ rows, $s$ columns and entries from a field $K$.
Also, $M_r(K)=M_{r,r}(K)$ and $GL_r(K)$ denotes as usual the group of invertible matrices of size $r$.

This notion was discovered and rediscovered in several contexts including algebraic complexity theory (in Strassen's work on tensor rank~\cite{strassen83}, see also~\cite{koiran20commuting} and the references therein),
in numerical analysis for the construction of cubature formulas~\cite{degani05} and in  quantum physics for the study of computational methods~\cite{degani06} or the study of the so-called "quantum Zeno dynamics"~\cite{burgarth14,orsucci15}. 
Unsurprisingly, 
commuting extensions have also attracted the attention of the linear algebra community~\cite{KSW09,KW08}. The term {\em commuting extension} was apparently coined in~\cite{degani05} (in \cite{KSW09,KW08} the term  {\em commuting completion} is used instead).

Given a tuple $(A_1,\ldots,A_p)$ of matrices of size $n$ and an integer $r \geq n$ we would like to know if a commuting extension $(Z_1,\ldots,Z_p)$ of size $r$ exists, if it is unique and how to compute it efficiently.
Strictly speaking, commuting extensions are never unique~\cite{degani05}. For $M \in GL_{r-n}(K)$, consider indeed the map $\rho_M:M_r(K) \rightarrow M_r(K)$ which sends
$Z =
\begin{pmatrix}
A & B\\
C & D
\end{pmatrix}$ to
 \begin{equation} \label{eq:action}
 \rho_M(Z) =
 \begin{pmatrix}
 I_n & 0\\
 0 & M
 \end{pmatrix}^{-1} Z \begin{pmatrix}
 I_n & 0\\
 0 & M
 \end{pmatrix} =
\begin{pmatrix}
A & BM\\
M^{-1}C & M^{-1}DM
\end{pmatrix}\end{equation}
where $I_n$ denotes the identity matrix of size $n$.
If $(Z_1,\ldots,Z_p)$ is a commuting extension of $(A_1,\ldots,A_p)$ then so is $(\rho_M(Z_1),\ldots,\rho_M(Z_p))$.
This follows immediately from the identity  
$\rho_M(ZZ')=\rho_M(Z)\rho_M(Z')$.
Let us say that a commuting extension of size~$r$ is {\em essentially unique} if it is unique up to this $GL_{r-n}$  action.

\subsection{Our results} 

In this paper we present three types of results:
\begin{itemize}
\item[(i)] Theorems on the essential uniqueness of commuting extensions for three matrices or more.
\item[(ii)] Algorithms for the computation of commuting extensions of minimal size. These algorithms work under the same assumptions
as our uniqueness theorems. They are applicable up to $r=4n/3$, and are apparently the first provably efficient algorithms for this problem going beyond $r=n+1$.
\item[(iii)] A "genericity theorem" on the applicability of our algorithms and uniqueness theorems.
\end{itemize}
In this introduction we give a precise statement only  for the uniqueness theorem for 3 matrices, which takes a fairly simple form:
\begin{theorem} \label{th:3unique}
Consider a tuple  $(A_1,A_2,A_3)$ of matrices of size~$n$ with entries in a field $K$ such that:
\begin{itemize}
\item[(i)] The three linear spaces $\Ima [A_1,A_2]$, $\Ima [A_1,A_3]$, $\Ima[A_2,A_3]$ are of dimension $2(r-n)$.
\item[(ii)] The three linear spaces $\Ima [A_1,A_2]+\Ima [A_1,A_3]$, $\Ima [A_2,A_1]+\Ima [A_2,A_3]$ and 
$\Ima [A_3,A_1]+\Ima [A_3,A_2]$ are of dimension $3(r-n)$.
\end{itemize}
The tuple $(A_1,A_2,A_3)$ does not have any commuting extension of size less than $r$.
If $(A_1,A_2,A_3)$ has a commuting extension of size $r$, it is essentially unique. Moreover, if a commuting extension of size $r$ exists
in the algebraic closure $\overline{K}$,  there is already one in the ground field $K$.  
\end{theorem}
Here we use the standard notation $[A,B ]=AB-BA$ for the commutator of two matrices, and we denote by $\Ima A$ the linear span of the columns of~$A$ (hence $\dim(\Ima A) = \rk A$). Note that the second hypothesis in this theorem can only be satisfied
when $r \leq 4n/3$. This restriction applies to all of our main results.

 It was pointed out in~\cite{degani05} that the authors "have no way of determining how many distinct families of commuting extensions of a given dimension exist" (by family, they mean
an equivalence class of tuples under the $GL_{r-n}$ action~(\ref{eq:action})). Theorem~\ref{th:3unique} provides a partial 
solution to this problem since it pinpoints an easily checkable criterion implying the existence of a single family at most.
It is also stated in~\cite{degani05} that they have had  only very limited with algorithms for computing commuting extensions.
The algorithms presented in this paper should help improve this situation as well.
These algorithms (for 3 matrices or more) can be viewed as constructive versions of the proof of Theorem~\ref{th:3unique}.

At this stage it is perhaps not  clear that  Theorem~\ref{th:3unique} has any interesting application, i.e., that there actually are tuples 
$(A_1,A_2,A_3)$  that satisfy the hypotheses of the theorem and admit commuting extensions of size $r$. 
We will see in Section~\ref{sec:generic} that if $K$ is an infinite field, such examples exist for all~$n$ and all $r \in [n,4n/3]$.\footnote{Such examples also exist in finite fields when $|K|$ is large enough compared to $n$; see Remark~\ref{rem:finite} in Section~\ref{sec:generic}.}
 Moreover, we will show that 
the situation studied in  Theorem~\ref{th:3unique} is in a precise sense (defined at the beginning of Section~\ref{sec:generic})
 the generic one. We would like to point out right away
that we do {\em not} mean by this that we pick generic matrices $A_1,A_2,A_3 \in M_n(K)$. This would not be an interesting model to study
because, as explained in Section~\ref{sec:nonunique}, for the range of values of $r$ considered in this paper ($r \leq 4n/3$) most triples of matrices of size $n$ 
do not have any commuting extension of size $r$. 

The worst-case complexity of computing commuting extensions seems to be unknown. For instance, it is not known whether
computing commuting extensions of a given size is NP-hard.

\subsection{Methods}

Our methods are almost entirely linear algebraic in nature (we also use some basic properties of Zariski topology in Section~\ref{sec:generic} to state and prove our genericity theorem).
In a nutshell, assumptions (i) and (ii) in Theorem~\ref{th:3unique} imply certain direct sum decompositions for the corresponding subspaces (see Lemma~\ref{lem:2direct} and Proposition~\ref{prop:3direct} in Section~\ref{sec:com}).
These decompositions in turn help us identify the unknown blocks in~(\ref{eq:block}).
For instance, we show  in Corollary~\ref{cor:topright} that $\Ima(B_1) = \Ima  [A_1,A_2]  \cap \Ima  [A_1,A_3].$
This does not identify $B_1$ uniquely, but this ambiguity is intrinsic to the problem 
due to the $GL_{r-n}$ action~(\ref{eq:action}).
Once a choice for $B_1$ is made, it turns out that all the blocks in the extension can be identified by solving a sequence
of linear systems.

For the  proof of our genericity theorem we assume that $K$ is an infinite field and we proceed in reverse: we first show that the above direct sum decompositions
hold generically. Then we show that the hypotheses of Theorem~\ref{th:3unique} follow from these direct sum decompositions.
This proof is partially nonconstructive: it shows that most inputs of a certain natural form (arising from simultaneously diagonalisable matrices) satisfy the hypotheses of Theorem~\ref{th:3unique}, but it does not give a way to explicitly write
down such inputs.

\subsection{Where to look for a commuting extension}

A thorough study of the commuting extension problem in the class of real symmetric matrices for $r=n+1$ and $p=2$ 
is presented in~\cite[Theorem 4.2]{KSW09}.  Their construction relies on the resolution of polynomial equations (including eigenvalue computations). In the same vein, one finds in~\cite[Example 4.3]{KW08} an explicit construction of a commuting extension
 of size $r=3$ for a pair of matrices of size $n=2$. These matrices have rational entries, but 
the extension has entries in $\mathbb{Q}[\sqrt{2}]$. 
This example, which is also interesting for its failure of essential uniqueness, is reproduced in Section~\ref{sec:nonunique}.
 By contrast, the algorithm presented in the present paper relies only on the resolution of linear systems. This is the reason why
 one can look without loss of generality for a solution in the ground field, as stated in the last part of Theorem~\ref{th:3unique}.

 \subsection{Further work}
 
 As explained earlier, the notion of commuting extensions is relevant in diverse areas such as algebraic complexity, numerical 
 analysis and quantum information. We will give in a companion paper some applications of the results presented here.
 One insight of the present paper is that under certain assumptions, it is easier to compute commuting extensions for
 three or more matrices than for two matrices (in the sense that no such algorithms or uniqueness results are known for two matrices).  It would certainly be interesting to obtain similar results for two matrices.
 At the moment, this problem is wide open. Also, it would be interesting to extend our results beyond $r=4n/3$.
 As explained in Section~\ref{sec:nonunique}, new difficulties seem to arise for $r=2n$ and beyond.
 
 \subsection{Organization of the paper}
 
 In the next section we present some general facts on the existence of commuting extensions and their uniqueness (or lack thereof).
 The main developments begin in Section~\ref{sec:unique}.  We prove there  our results on the essential uniqueness of commuting extensions, for 3 matrices and then more generally for $p\geq 3$ matrices. In Section~\ref{sec:unique} and throughout the paper, 
 most of the important ideas already appear for $p=3$ matrices.
 We build on these results in Section~\ref{sec:algo} to derive algorithms for the construction of commuting extensions.
 Finally, it is shown in Section~\ref{sec:generic} that our results are applicable to a wide range of input matrices.
  
 \section{Existence and non-uniqueness of commuting extensions} \label{sec:nonunique}
 
 Commuting extensions exist for any tuple $(A_1,\ldots,A_p)$ of matrices of size $n$ if we allow the size of the extended matrices to be large enough compared to~$n$ and $p$. An explicit construction of size $r=pn$ based on "block circulants" 
 was presented in~\cite{degani05}. As already pointed out in~\cite{koiran20commuting}, commuting extensions of size $r=2n$  exist for any tuple of matrices.
 Namely, one can  take:
\begin{equation} \label{eq:2n}
N_i = \begin{pmatrix}
A_i & -A_i\\
A_i & -A_i
\end{pmatrix}
\end{equation}
since $N_i N_j=0$ for all $i,j$. This simple construction is also relevant  for the study of the (essential) uniqueness of commuting 
extensions. Anticipating on Section~\ref{sec:generic}, let us pick an arbitrary tuple $(Z_1,\ldots,Z_p)$ of simultaneously 
diagonalisable matrices of size $r$. For $n<r$, these matrices form a commuting extension of $(A_1,\ldots,A_p)$ 
where $A_i$ is the top left block of $Z_i$ of size $n$. If $r=2n$, $(A_1,\ldots,A_p)$ also admits another commuting
extension of size $2n$, namely, the extension $(N_1,\ldots,N_p)$. 
This extension is certainly not equivalent\footnote{except if $Z_i=0$ for all $i$.}  to
$(Z_1,\ldots,Z_p)$ under the $GL_{r-n}$ action~(\ref{eq:action})
 because the $Z_i$ are diagonalisable and the $N_i$ are nilpotent.
 This construction shows that the extension of our uniqueness results from $r=4n/3$ to $r=2n$ and beyond is problematic.
 One can perhaps obtain similar uniqueness theorems beyond $2n$ by considering only commuting extensions that are diagonalisable, but this problem is wide open at the moment.
Essential uniqueness also fails in the following situation:
\begin{example}
Let $A=\begin{pmatrix}
1 & 0\\
0 & 2
\end{pmatrix}$,
$B=\begin{pmatrix}
0 & 1\\
1 & 0
\end{pmatrix}$.
The commuting extension 
$$A_{ext}=\begin{pmatrix}
1 & 0 & 2\sqrt{2}/x \\
0 & 2 & -4/x \\
\frac{-y}{\sqrt{2}} & y & -2y/x
\end{pmatrix},
B_{ext}=\begin{pmatrix}
0 & 1 & -2/y\\
1 & 0 & \sqrt{2}/y\\
\frac{x}{2} & \frac{-x}{2\sqrt{2}} & \frac{x}{\sqrt{2}y}
\end{pmatrix}.$$
is constructed in~\cite[Example 4.3]{KW08}. Here $x$ and $y$ are arbitrary. Note that there is one additional degree of freedom (choice of $x$ {\em and} choice of $y$) compared to~(\ref{eq:action}).
It can be checked that 
$A_{ext}B_{ext}= \alpha I_3$ for some $\alpha \neq 0$. Any pair of matrices satisfying such a relation must commute, 
and this is the way this example is constructed in~\cite{KW08}.
\end{example}
More generally, it follows from~\cite[Theorem 4.1]{KW08} that two invertible matrices of size $n$ always have a  commuting extension of size $2n-1$. This results lends itself well to an efficient construction~\cite[Algorithm 4.2]{KW08}, but it does not allow 
the determination of a commuting extension of minimal size. By contrast, Theorem~\ref{th:3unique} certifies that the commuting extensions constructed in the present paper are of minimal size.

 Finally, we point out that when $K$ is an infinite field,\footnote{or if $K$ is finite but large enough.} most tuples $(A_1,\ldots,A_p) \in M_n(K)^p$ do not have any commuting extension of size $r$
  for the range of values of $r$ considered in this paper ($r \leq 4n/3$). 
  This is so because  commuting extensions of size~$r$ do not even exist for most pairs $(A_1,A_2) \in M_n(K)^2$ if $n$ is large enough ($n \geq 6$ suffices).
  This can be seen by comparing the dimension of $M_n(K)^2$ (namely, $2n^2$) to the (smaller) dimension of the set of pairs of commuting
  matrices of size $r$, which is equal to $r^2+r$.
  The latter result follows from two facts:
  \begin{itemize}
  \item[(i)] A pair of of matrices of size $r$ commutes if and only if it is in the closure of the set of pairs of matrices of size $r$  that are simultaneously diagonalisable.
  \item[(ii)] The closure of the set of $p$-tuples of matrices of size $r$ that are simultaneously diagonalisable has dimension $r^2+(p-1)r$. 
  \end{itemize}
  Here we work with the standard notion of dimension from algebraic geometry
   (computed over the algebraic closure $\overline{K}$).
  Fact (i) was originally established by Motzkin and Taussky~\cite[Theorem 6]{motzkin55} (see also Theorem~6.8.1 in~\cite{omeara11}),
  and Fact (ii) by Guralnick and Sethuraman~\cite[Proposition~6]{guralnick00} as pointed out in~\cite[Section 3.4]{jelisiejew22}.
  
 \section{Uniqueness Theorems} \label{sec:unique}
 
 In Section~\ref{sec:com} we work out some basic properties of commutators which will be useful throughout the paper.
 Theorem~\ref{th:3unique} is proved in Section~\ref{sec:3unique}, and we extend it to $p \geq 3$ matrices 
 in Section~\ref{sec:4unique}.
 
\subsection{Commutators} \label{sec:com}

Throughout Section~\ref{sec:com} we fix a tuple of matrices $(A_1,\ldots,A_p)$ having  a commuting extension of size $r$,
which we denote as usual by $(Z_1,\ldots,Z_p)$. 
We keep the same notation as in~(\ref{eq:block}) for the block decomposition of the $Z_i$.
\begin{lemma} \label{lem:2direct}
For any pair of matrices $(A_k,A_l)$ in the tuple $(A_1,\ldots,A_p)$ we have
\begin{equation} \label{eq:commutator}
[A_k,A_l]=B_l C_k - B_k C_l
\end{equation}
and $\rk [A_k,A_l] \leq 2(r-n)$. 

Assume furthermore that  $\rk [A_k,A_l] = 2(r-n)$.
Then $r \leq 3n/2$, the matrices $B_k,B_l,C_k,C_l$ are all of rank $r-n$  and:
\begin{equation} \label{eq:direct}
\Ima  [A_k,A_l]  = \Ima(B_k) \oplus \Ima(B_l),
\end{equation}
\begin{equation} \label{eq:direct2}
\Ima  [A_k,A_l]^T  = \Ima(C_k^T) \oplus \Ima(C_l^T).
\end{equation}
\end{lemma}
\begin{proof}
In order to obtain the expression for $[A_k,A_l]$  write $Z_k Z_l$ and $Z_l Z_k$ as two block matrices;
then  equate the top-left blocks of these matrices. Note that  $B_k$ and $B_l$  are of rank at most $r-n$ since these two matrices have $r-n$ columns. 
This implies $\rk [A_k,A_l] \leq 2(r-n)$ since~(\ref{eq:commutator}) expresses  $[A_k,A_l]$ 
as the difference of two matrices of rank at most $r-n$. This expression also implies that 
\begin{equation} \label{eq:sumincl}
\Ima  [A_k,A_l]  \subseteq  \Ima(B_k) + \Ima(B_l).
\end{equation}
The bound $r \leq 3n/2$ follows immediately from  $\rk [A_k,A_l] = 2(r-n)$. This assumption also implies~(\ref{eq:direct}) 
since the two subspaces on the right-hand side of~(\ref{eq:sumincl}) are of dimension at most $r-n$, and it also implies
$\rk B_k = \rk B_l = r-n$.
Finally, $\rk C_k = \rk C_l = r-n$ and the direct sum decomposition~(\ref{eq:direct2}) follow from a similar reasoning applied to the transpose of~(\ref{eq:commutator}).
\end{proof}
 The bound $\rk [A_k,A_l] \leq 2(r-n)$ in this lemma is a basic property of commuting extensions used by 
 Strassen~\cite{strassen83} in his lower bound for (border) tensor rank. This inequality was rediscovered in~\cite{degani05,orsucci15}.
\begin{proposition} \label{prop:3direct}
For any triple of matrices  $(A_k,A_l,A_m)$  in the tuple $(A_1,\ldots,A_p)$, the linear space
 $V_{klm}= \Ima [A_k,A_l]+\Ima [A_k,A_m]$ satisfies $\dim  V_{klm} \leq 3(r-n)$. 
If this inequality is an equality then $r \leq 4n/3$, $B_k,B_l$ and $B_m$ are of rank $r-n$ and we have the direct sum decomposition
\begin{equation} \label{eq:3direct}
V_{klm}= \Ima(B_k) \oplus \Ima(B_l) \oplus \Ima(B_m).
\end{equation}
\end{proposition}
\begin{proof}
Note that
\begin{equation} \label{eq:3sumincl}
V_{klm} \subseteq \Ima(B_k) + \Ima(B_l) + \Ima(B_m)
\end{equation}
 by~(\ref{eq:sumincl})
and by the similar inclusion $\Ima  [A_k,A_m]  \subseteq  \Ima(B_k) + \Ima(B_m).$
This implies $\dim  V_{klm} \leq 3(r-n)$ 
since this linear space is included in the sum of 3 spaces of dimension at most $r-n$ each.

Let us now assume  that $\dim  V_{klm} = 3(r-n)$. This is only possible when $r \leq 4n/3$ since $V_{klm} \subseteq K^n$.
Moreover, the 3 subspaces on the right-hand side of~(\ref{eq:3sumincl}) must now be in direct sum since they
are of dimension at most $r-n$ each, and they must in fact be of dimension exactly $r-n$.  
\end{proof}
As an important corollary, under the assumptions of Lemma~\ref{lem:2direct} and Proposition~\ref{prop:3direct}
 the upper right block $B_k$ has the same image
in all commuting extensions of $(A_1,\ldots,A_p)$.
\begin{corollary} \label{cor:topright}
If $\dim (\Ima [A_k,A_l]+\Ima [A_k,A_m]) = 3(r-n)$, $\dim \Ima [A_k,A_l] = 2(r-n)$ and $\dim \Ima [A_k,A_m] = 2(r-n)$ 
we have $$\Ima(B_k) = \Ima  [A_k,A_l]  \cap \Ima  [A_k,A_m].$$
\end{corollary}
\begin{proof}
This follows from~(\ref{eq:3direct}), from~(\ref{eq:direct}) and from the similar direct sum decomposition 
$\Ima  [A_k,A_m]  = \Ima(B_k) \oplus \Ima(B_m).$
\end{proof}
So far, we have only exploited the equality of the top left blocks in $Z_k Z_l$ and $Z_l Z_k$.
By equating the 
top right blocks we obtain:
\begin{lemma} \label{lem:bottomright}
For any pair of indices $k,l \in \{1,\ldots,p\}$ we have
$$B_l D_k - B_k D_l = A_k B_l - A_l B_k.$$
\end{lemma}
This lemma will give us a way of determining the bottom right blocks of a commuting extension once we have determined the top right blocks.

\subsection{Uniqueness for 3 matrices: proof of Theorem~\ref{th:3unique}} 
\label{sec:3unique}

Consider a tuple $(A_1,A_2,A_3)$ of matrices of size~$n$ satisfying the hypotheses of Theorem~\ref{th:3unique}.
If there is a commuting extension of size less than $r$, we can obtain a commuting extension of size exactly $r$ by 
adding rows and columns of 0's to the three matrices in the extension.
 The top-right blocks in this extension would be of rank $<r-n$,
 in contradiction with Lemma~\ref{lem:2direct} (and also with Proposition~\ref{prop:3direct}). This establishes the first part of Theorem~\ref{th:3unique} (one could also argue directly that the existence of a commuting extension of size $s<r$
  implies $\rk [A_1,A_2] \leq 2(s-n) < 2(r-n)$, in contradiction with hypothesis (i) of Theorem~\ref{th:3unique}).

Suppose now that our tuple has two commuting extensions of size $r$, $(Z_1,Z_2,Z_3)$ and $(Z'_1,Z'_2,Z'_3)$.
In order to prove Theorem~\ref{th:3unique}, we must show that $Z'_i = \rho_M(Z_i)$ for some $M \in GL_{r-n}(K)$.
By Corollary~\ref{cor:topright} the top right blocks of $Z_1$ and $Z'_1$ have the same image.
 Hence there exists $M \ \in GL_{r-n}(K)$ such that $B'_1=B_1M$. 
 By applying $\rho_{M^{-1}}$ to the $Z'_i$ we can reduce to the case $B'_1=B_1$.
  The essential uniqueness in Theorem~\ref{th:3unique} therefore follows from the following result.
 \begin{theorem} \label{th:topright}
Consider a tuple $(A_1,A_2,A_3)$ of matrices of size~$n$ such that:
\begin{itemize}
\item[(i)] The three linear spaces $\Ima [A_1,A_2]$, $\Ima [A_1,A_3]$, $\Ima[A_2,A_3]$ are of dimension $2(r-n)$.
\item[(ii)] The three linear spaces $\Ima [A_1,A_2]+\Ima [A_1,A_3]$, $\Ima [A_2,A_1]+\Ima [A_2,A_3]$ and 
$\Ima [A_3,A_1]+\Ima [A_3,A_2]$ are of dimension $3(r-n)$.
\end{itemize}
If $(Z_1,Z_2,Z_3)$ and $(Z'_1,Z'_2,Z'_3)$ are two commuting extensions of size $r$ such that $B_1=B'_1$,
these extensions are identical: $(Z_1,Z_2,Z_3)=(Z'_1,Z'_2,Z'_3)$.
 \end{theorem}
 We will repeatedly use the following simple lemma in the proof of Theorem~\ref{th:topright}.
 \begin{lemma} \label{lem:sum}
Suppose that $P,P',Q,Q' \in M_n(K)$ are four matrices such that $P+Q=P'+Q'$, $\Ima(P)=\Ima(P')$ and $\Ima(Q)=\Ima(Q')$.
If $\Ima(P)$ and $\Ima(Q)$ are in direct sum then $P=P'$ and $Q=Q'$.
\end{lemma}
\begin{proof}
Let $\cal E$ (respectively, $\cal F$) be the set of matrices $M$ such that $\Ima(M) \subseteq  E$ 
(resp.,  $\Ima(M) \subseteq  F$) where $ E=\Ima(P)=\Ima(P')$ and ${ F}=\Ima(Q)=\Ima(Q')$.
Clearly, $\cal E$ and  $\cal F$ are linear subspaces of $M_n(K)$, and they are in direct sum:
if $M \in {\cal E} \cap {\cal F}$ then $\Ima(M) \subseteq E \cap F = \{0\}$, hence  $M=0$.
To conclude, note that $P+Q=P'+Q'$ where $P,P' \in {\cal E}$ and $Q, Q' \in {\cal F}$. Hence $P=P'$ and $Q=Q'$ by the
direct sum property.
\end{proof}
A variation on this is:
\begin{lemma} \label{lem:span}
Let $P,Q \in M_{n,r-n}(K)$ be two matrices of rank $r-n$ with $\Ima(P)$ and $\Ima(Q)$ in direct sum.
For any matrix $T \in M_{n,r-n}(K)$, there  exists a pair of matrices $X,Y \in M_{r-n}(K)$ such that $T=PX+QY$
if and only if $\Ima(T) \subseteq \Ima(P) \oplus \Ima(Q)$, and in this case the pair $(X,Y)$ is unique.
\end{lemma}
\begin{proof}
Let $p_k,q_k$ be the columns of $P$ and $Q$. Solving the system $T=PX+QY$ amounts to writing each column $t_j$ of~$T$
as a linear combination $t_j=\sum_k x_{kj} p_k + \sum_k y_{kj} q_k$. 
This is possible if and only if  $\Ima(T) \subseteq \Ima(P) \oplus \Ima(Q)$. In this case, the solution is unique since the columns
$p_k$ and $q_k$ form a basis of $\Ima(P) \oplus \Ima(Q)$. 
\end{proof}
\begin{proof}[Proof of Theorem~\ref{th:topright}]
We will show that the 12 blocks appearing in the block decomposition of $Z_1,Z_2,Z_3$ are equal to their counterparts
in $Z'_1,Z'_2,Z'_3$. We already know that the 3 top left blocks are equal (by definition of commuting extensions), and 
that $B_1=B'_1$ by hypothesis. It therefore remains to prove 8 equalities. Note that the off-diagonal blocks in the two
commuting extensions are all of dimension $r-n$ by Lemma~\ref{lem:2direct}.

Our first goal is to prove that $C_2=C'_2$. By~(\ref{eq:commutator}) we have
\begin{equation} \label{eq:12}
[A_1,A_2]=B_2 C_1 - B_1 C_2 = B'_2 C'_1 - B'_1 C'_2.
\end{equation}
The matrices appearing in these two decompositions of $[A_1,A_2]$ are all of full rank ($r-n$),
so $\Ima(B_2 C_1)=\Ima(B_2)$, $\Ima(B'_2 C'_1) = \Ima(B'_2)$, $\Ima(B_1 C_2) = \Ima(B_1)$, 
and $\Ima(B'_1 C'_2)=\Ima(B'_1)=\Ima(B_1)$.
Moreover, $\Ima(B_2)=\Ima(B'_2)$ by Corollary~\ref{cor:topright} applied to the triple $(A_2,A_1,A_3)$.
Since the images of $B_1$ and $B_2$ are in direct sum, it follows from  Lemma~\ref{lem:sum} that $B_2C_1=B'_2C'_1$ 
and $B_1C_2=B'_1C'_2=B_1C'_2$. Since $B_1$ is of full column rank, $B_1C_2=B_1 C'_2$ implies $C_2=C'_2$ as desired.

Then we use the knowledge that $C_2=C'_2$ to obtain the equality $B_3=B'_3$. This is very similar to the above derivation 
of $C_2=C'_2$, but instead of~(\ref{eq:12}) we start from
\begin{equation} \label{eq:23}
[A_2,A_3]=B_3 C_2 - B_2 C_3 = B'_3 C'_2 - B'_2 C'_3.
\end{equation}
By Corollary~\ref{cor:topright} applied to the triple $(A_3,A_1,A_2)$, $\Ima(B_3)=\Ima(B'_3)$ and we have already seen
that $\Ima(B_2)=\Ima(B'_2)$. Hence the two equalities $B_3C_2 = B'_3C'_2=B'_3C_2$, $B_2C_3 = B'_2 C'_3$ follow from another
application of Lemma~\ref{lem:sum}. Since $C_2$ is of full row rank, the equality $B_3C_2 =B'_3C_2$ implies $B_3=B'_3$.
Next, we derive from this equality the new equality $C_1=C'_1$. This is entirely parallel to the derivation of $C_2=C'_2$
from $B_1=B'_1$, but instead of~(\ref{eq:12}) we start from
\begin{equation} \label{eq:13}
[A_1,A_3]=B_3 C_1 - B_1 C_3 = B'_3 C'_1- B'_1 C'_3.
\end{equation}
In addition to $C_1=C'_1$, we also obtain as before an additional equality: ($B_1C_3=B'_1C'_3$) from Lemma~\ref{lem:sum}.

At this stage we have proved 3 out of the 8 block equalities. Moreover, we have just shown that $B_1C_3=B'_1C'_3$. 
Since $B_1=B'_1$, this implies $C_3=C'_3$. Recall now that we have obtained $B_2C_3=B'_2C'_3$ in the paragraph
following~(\ref{eq:23}). Armed with our new knowledge that $C_3=C'_3$, we can derive from this the equality $B_2=B'_2$.

We now have 5 out of the 8 block equalities, and it just remains to prove the 3 equalities of bottom-right blocks.
For this we use Lemma~\ref{lem:bottomright} and Lemma~\ref{lem:span}. Let us fix two distinct indices $k,l \in \{1,2,3\}$.
Since the top right blocks are identical in the two commuting extensions, it follows from Lemma~\ref{lem:bottomright} that
each of the pairs $(D_k,D_l)$ and $(D'_k,D'_l)$ is a solution to the linear system 
$$B_l X - B_k Y= A_k B_l - A_l B_k.$$
Lemma~\ref{lem:span} shows that the solution $(X,Y)$ is unique since $B_k,B_l$  are of rank $r-n$ and have their images in direct sum. This completes the proof of Theorem~\ref{th:topright}. 
\end{proof} 
The proof of Theorem~\ref{th:3unique} is also
complete now, except for the last part: if a commuting extension of size $r$ exists
in the algebraic closure $\overline{K}$,  there is already one in the ground field $K$. 
This follows from Remark~\ref{rem:field} at the end of Section~\ref{sec:3algo}.

\subsection{Uniqueness for more than 3 matrices} \label{sec:4unique}

Let us fix 3 distinct indices $k,l,m \leq p$.
Say that a tuple of matrices $(A_1,\ldots,A_p)$ of matrices of size $n$ satisfies hypothesis $(H_{klm})$
 if the triple $(A_k,A_l,A_m)$ satisfies the hypotheses of Theorem~\ref{th:3unique}, i.e., 
\begin{itemize}
\item[(i)] The three linear spaces $\Ima [A_k,A_l]$, $\Ima [A_k,A_m]$, $\Ima[A_l,A_m]$ are of dimension $2(r-n)$.
\item[(ii)] The three linear spaces $\Ima [A_k,A_l]+\Ima [A_k,A_m]$, $\Ima [A_l,A_k]+\Ima [A_l,A_m]$ and 
$\Ima [A_m,A_k]+\Ima [A_m,A_l]$ are of dimension $3(r-n)$.
\end{itemize}
Uniqueness results for more than 3 matrices can be obtained under various combinations of the $(H_{klm})$.
For instance: 
\begin{theorem} \label{th:unique}
Consider a tuple of matrices $(A_1,\ldots,A_p)$ of matrices of size~$n$ with $p \geq 3$ such that for all $2 \leq l \leq p$ there is some $ m \not \in \{1,l\}$ satisfying hypothesis $(H_{1lm})$. 
If $(A_1,\ldots,A_p)$ has a commuting extension of size $r$, it is essentially unique. 
Moreover, if a commuting extension of size $r$ exists
in the algebraic closure $\overline{K}$,  there is already one in the ground field $K$.
\end{theorem}
For $p=3$, this theorem is equivalent to Theorem~\ref{th:3unique}.
\begin{proof}
Suppose that $(Z_1,\ldots,Z_p)$ and $(Z'_1,\ldots,Z'_p)$  are two commuting extensions of size $r$.
We need to show that $Z'_i = \rho_M(Z_i)$ for some $M \in GL_{r-n}(K)$.
We first apply the  argument appearing at the beginning of Section~\ref{sec:3unique}.
Indeed, by hypothesis  there exists some $m \geq 3$ such that $(H_{12m})$ holds.
 By Corollary~\ref{cor:topright} the top right blocks of $Z_1$ and $Z'_1$ have the same image.
 Hence there exists $M \ \in GL_{r-n}(K)$ such that $B'_1=B_1M$.\footnote{In fact, Theorem~\ref{th:3unique} applied
 to $(A_1,A_2,A_m)$ shows that $Z'_1=\rho_M(Z_1), Z'_2=\rho_M(Z_2), Z'_m=\rho_M(Z_m)$.}
 Applying $\rho_{M^{-1}}$ to the $Z'_1,\ldots,Z'_p$, we  reduce to the case $B'_1=B_1$.
To conclude, we use this block identity to show that $Z'_l=Z_l$ for all $l \in \{1,\ldots,p\}$. 
Indeed, for any~$l \geq 2$ there is some $m$ such that $(H_{1lm})$ holds. 
Since $B_1=B'_1$, it follows from Theorem~\ref{th:topright} applied to $(A_1,A_l,A_m)$ that $Z'_l=Z_l$ and $Z'_m=Z_m$.
Finally, as in Theorem~\ref{th:3unique} the last part of Theorem~\ref{th:unique} follows from Remark~\ref{rem:field}. 
\end{proof}
We will see in Section~\ref{sec:generic} that the generic situation is that  {\em all} of the hypotheses
$(H_{klm})$ hold simultaneously.

\section{Computing the extensions} \label{sec:algo}

\subsection{An algorithm for three matrices} \label{sec:3algo}

In this section we present and analyze an algorithm for the computation of commuting extensions of size $r$ of a triple of matrices 
$(A_1,A_2,A_3)$ of size~$n$ satisfying the hypotheses of Theorem~\ref{th:3unique} 
(in particular, we must have $r \leq 4n/3$). It can be viewed as an algorithmic version of the
proof of that theorem, and attempts to determine the 9 unknown blocks $B_i,C_i,D_i$ in the extension in the same order as in that proof.
We present this algorithm as a sequence of 9 steps:
\begin{enumerate}
\item Compute an arbitrary basis of $\Ima [A_1,A_2] \cap \Ima [A_1,A_3]$.\\ Let $B_1 \in M_{r,r-n}(K)$ be the matrix
whose column vectors are the elements of this basis.

\item Let $V_2 = \Ima [A_2,A_1] \cap \Ima [A_2,A_3]$. Compute an arbitrary basis of $V_2$, use it to write 
$[A_1,A_2] = M_1 - M_2$ where $\Ima(M_1) \subseteq  \Ima(B_1)$ and $\Ima(M_2) \subseteq V_2$.
Then compute a matrix $C_2$ such that $B_1C_2=M_2$.
\item  Let $V_3 = \Ima [A_3,A_1] \cap \Ima [A_3,A_2]$. Compute an arbitrary basis of $V_3$, use it to write 
$[A_2,A_3] = N_3 - N_2$ where $\Ima(N_3) \subseteq  V_3$ and $\Ima(N_2) \subseteq V_2$.
Then compute a matrix $B_3$ such that $B_3C_2=N_3$.
\item Write $[A_1,A_3]=P_3-P_1$ where $\Ima(P_3) \subseteq V_3$, $\Ima(P_1) \subseteq \Ima(B_1)$. 
Then compute a matrix $C_1$ such that $P_3=B_3C_1$.

\item Compute a matrix $C_3$ such that $P_1=B_1C_3$.

\item Compute a matrix $B_2$ such that $N_2=B_2C_3$.

\item Solve the linear system $B_3X-B_2Y=A_2B_3-A_3B_2$, then set $D_2=X,D_3=Y$.

\item Compute a matrix $D_1$ such that $B_2D_1-B_1D_2=A_1B_2-A_2B_1$.

\item Arrange the blocks $A_i,B_i,C_i,D_i$ in three matrices $Z_1,Z_2,Z_3$. If the $Z_i$ do not commute, reject. 
If they commute, output $(Z_1,Z_2,Z_3)$.
\end{enumerate}
If at any time in the above algorithm one of the linear systems cannot be solved, we reject. Also, if at any time in steps 2 to 8
the solution found for one of the blocks $B_i,C_i,D_i$ is not unique, we reject (note however that the solution for $B_1$
at Step 1 is not unique).

\begin{theorem} \label{th:3algo}
Let $(A_1,A_2,A_3)$ be a triple of matrices of size $n$ satisfying the hypotheses of Theorem~\ref{th:3unique}:
\begin{itemize}
\item[(i)] The three linear spaces $\Ima [A_1,A_2]$, $\Ima [A_1,A_3]$, $\Ima[A_2,A_3]$ are of dimension $2(r-n)$.
\item[(ii)] The three linear spaces $\Ima [A_1,A_2]+\Ima [A_1,A_3]$, $\Ima [A_2,A_1]+\Ima [A_2,A_3]$ and 
$\Ima [A_3,A_1]+\Ima [A_3,A_2]$ are of dimension $3(r-n)$.
\end{itemize}
The above algorithm runs in polynomial time and outputs  a commuting extension of size $r$  of $(A_1,A_2,A_3)$ 
if such an extension exists. Otherwise the algorithm rejects.
\end{theorem}
By "polynomial time", we mean that the number of arithmetic operations and equality tests between elements of 
$K$  performed by the algorithm is polynomially bounded in $n$. Moreover, it can be shown that for $K=\mathbb{Q}$ the bit size of the numbers
involved in the computation remains polynomially bounded (in the bit size of the input).  This is due to the standard fact 
that the solutions of linear systems have a well-controlled bit size (we omit the details).
\begin{proof}[Proof of Theorem~\ref{th:3algo}]
By Step 9, the algorithm reject its input if it does not have any commuting extensions of size $r$.
Suppose now that a commuting extension $(Z_1,Z_2,Z_3)$  exists. We first follow the argument at the beginning 
of Section~\ref{sec:3unique}. 
By Corollary~\ref{cor:topright}, the top right block $B_1$ of $Z_1$ must satisfy $\Ima(B_1) = \Ima [A_1,A_2] \cap \Ima [A_1,A_3]$, and this image must be of dimension $r-n$
by Lemma~\ref{lem:2direct}. Moreover, by~(\ref{eq:action}), for any matrix $B'_1 \in M_{n,r-n}(K)$ having the same image 
there exists a commuting extension $(Z'_1,Z'_2,Z'_3)$ of size $r$ where $B'_1$ is the top right block of $Z'_1$.
We can therefore assume without loss of generality that the top right block of $Z_1$ is the matrix chosen by the algorithm
at Step 1.

The correctness of Step 2 follows from the second paragraph of the proof of Theorem~\ref{th:topright}. In particular, 
$V_2 = \Ima(B_2)$ (but the algorithm has not determined $B_2$ yet). As explained in the proof of Lemma~\ref{lem:sum},
writing $[A_1,A_2] = M_1 - M_2$ amounts to the decomposition of  $[A_1,A_2]$ as a sum of two matrices belonging
to two subspaces in direct sum. The solution is therefore unique and can be found by solving a linear system.
Note that we will in fact have  $\Ima(M_1) =  \Ima(B_1)$, $\Ima(M_2) = V_2$, $M_1=B_2C_1$.

The analysis of Step 3 is very similar to the above analysis of Step 2; see also (\ref{eq:23}) and the surrounding paragraph.
Note in particular that $V_3 = \Ima(B_3)$, and we have $V_2 = \Ima(B_2)$ from the analysis of Step 2. Moreover,
we will have $N_2=B_2C_3$.

The correctness of Step 4 follows from~(\ref{eq:13}) and the surrounding paragraph;  in particular, we have $P_1=B_1C_3$.
This equality is of course the justification for Step 5, and the equality $N_2=B_2C_3$ derived in the analysis of Step 3 is the justification for Step 6. 
The justification for Steps 7 and 8 is in the last paragraph of the proof of Theorem~\ref{th:topright}. This completes
the correctness proof.

Finally, we note that the algorithm consists in the resolution
of some fixed number of linear systems (this number is independent of $n$, and the size of the systems is at most quadratic in $n$).
Each system is set up using the solutions found at the previous steps.
 This shows that the running time is polynomially bounded in $n$.
 \end{proof}

\begin{remark} \label{rem:field}
When $A_1,A_2,A_3$ have their entries in some field $K$, the extension $Z_1,Z_2,Z_3$ computed by this algorithm
has also its entries in $K$. This is due to the fact that these entries are computed from the input by a sequence
of arithmetic operations (in particular, we do not need to compute polynomial roots or construct field extensions).
The same remark applies to the algorithm of Section~\ref{sec:4algo}.
\end{remark}

\subsection{More than three matrices} \label{sec:4algo}

In this section we present and analyze an algorithm for the computation of commuting extensions of size $r$
 of a tuple of matrices 
$(A_1,\ldots,A_p)$ where $p \geq 4$. We make the same assumptions on this tuple as in our uniqueness theorem for
$p \geq 3$ matrices (Theorem~\ref{th:unique}). As before, this restricts us to the range $r \leq 4n/3$.
The algorithm goes as follows (recall that hypothesis $(H_{klm})$ is defined at the beginning of Section~\ref{sec:4unique}):
\begin{enumerate}
\item Find $m$ such that $(H_{12m})$ holds.

\item Run steps 1 to 8 of the algorithm of Section~\ref{sec:3algo} on input $(A_1,A_2,A_m)$ to compute 
a candidate extension $(Z_1,Z_2,Z_m)$.
\item For $l$ from $3$ to $p$, if $Z_l$ has not been computed yet do the following:
\begin{itemize}
\item[(a)] Find $m$ such that $(H_{1lm})$ holds.
\item[(b)]  Run steps 2 to 7 of the algorithm of Section~\ref{sec:3algo} on input $(A_1,A_l,A_m)$ to compute 
a candidate extension $(Z_1,Z_l,Z_m)$. 
\end{itemize}
\item Output $(Z_1,\ldots,Z_p)$ if these matrices pairwise commute, reject if they don't.
\end{enumerate}

\begin{theorem} \label{th:algo}
Let $(A_1,,\ldots,A_p)$ be a triple of matrices of size $n$  such that there is for all $2 \leq l \leq p$ some $ m \not \in \{1,l\}$ satisfying hypothesis  $(H_{1lm})$.
The above algorithm runs in polynomial time and outputs  a commuting extension of size $r$  of $(A_1,A_2,A_3)$ 
if such an extension exists. Otherwise the algorithm rejects.
\end{theorem}
\begin{proof}
By Step 4, the algorithm reject its input if it does not have any commuting extensions of size $r$. Due to 
this global commutativity test, we do not have  run step 9 of the algorithm of Section~\ref{sec:3algo} 
(we could of course still perform these local commutativity tests, and reject early if one of them fails).

Suppose now that a commuting extension $(Z'_1,\ldots,Z'_p)$  exists. In this case, correctness follows the proof of 
Theorem~\ref{th:unique}. More precisely, by applying some $\rho_{M^{-1}}$ to $Z'_1,\ldots,Z'_p$, 
we can assume that the algorithm finds $Z_1=Z'_1,Z_2=Z'_2,Z_m=Z'_m$ at Step 2.
It will then compute $Z_3=Z'_3,\ldots,Z_p=Z'_p$ at Step 3.
At Step 3.(b) we do not run steps 1 and 8 of the algorithm of Section~\ref{sec:3algo} because their purpose is to determine
the blocks $B_1$ and $D_1$, and the whole matrix $Z_1$ has already been determined at this stage. For the same reason, we do not need to recompute $C_1$ when we run step 4 of the algorithm of Section~\ref{sec:3algo}  (but we do need to compute $P_1$ for subsequent use at step 5).
Finally, we note that some matrices $Z_m$ will already be computed before the loop counter at Step~3 reaches the value 
$l=m$ (for instance, at Step 1 we compute  $Z_m$ for some $m \geq 3$). It is of course not necessary to compute these matrices again,
and we can proceed immediately to the next iteration.
\end{proof}
At Step 4 of the algorithm we check that $p(-1)/2$ pairs of matrices commute. As a side remark, note that if randomization is allowed one can  just compute two random linear combinations of the $Z_i$, and check that this single pair commutes. The (simple) analysis of this randomized test can be found in~\cite[Lemma 1.6]{KoiranSaha23}.

\section{Some Generic Considerations} \label{sec:generic}

In this section $K$ is an infinite field. 
We will construct  tuples of matrices that satisfy the hypotheses of our uniqueness theorems
 and admit commuting extensions of size $r$. 
 In the previous sections, we began with a tuple $(A_1,\ldots,A_p)$ of matrices of size $n$ and studied the properties of its extensions.
 Here we proceed in reverse: we'll start from commuting matrices $(Z_1,\ldots,Z_p)$ of size $r$ and will define $A_i$
 as the top left block of size $n$ of the~$Z_i$.
 To make sure that the $Z_i$ commute, we  take them to be simultaneously diagonalisable. There are other ways of 
 constructing commuting matrices\footnote{From the point of view of algebraic geometry, this is due to the fact that: (i) the set of $p$-tuples of simultaneously diagonalisable  matrices is not closed; (ii) the closure of this set is in general not the only irreducible 
 component of the variety of  $p$-tuples of commuting matrices. See~\cite{jelisiejew22} for some recent progress on the
 description of the irreducible components.}  (see for instance Section~\ref{sec:nonunique}), but diagonalisable extensions
 turn out to be of particular interest in numerical analysis~\cite{degani05}, in algebraic complexity (see~\cite{strassen83}
 and the other references in~\cite{koiran20commuting}) and in quantum physics~\cite{degani06,orsucci15}.

 We will therefore pick diagonal matrices $D_1,\ldots,D_p$ of size $r$, an invertible matrix $R$ of size $r$,
 and  will set  $Z_i=R^{-1}D_iR$. The main result of this section is as follows.
 \begin{theorem}[Genericity Theorem] \label{th:generic}
 Let $Z_i=R^{-1}D_iR$ where $R \in GL_r(K)$ and where  $D_1,\ldots,D_p$ are diagonal matrices of size $r$.
 Let $A_i$ be the top left block of size $n$ of $Z_i$.
 The two following properties hold for a generic choice of $R$ and of the $D_i$:
 \begin{itemize}
 \item[(i)] If $p \geq 2$ and $r \leq 3n/2$, $\dim \Ima  [A_k,A_l] = 2(r-n)$ for all $k \neq l$.
 \item[(ii)] If $p \geq 3$ and $r \leq 4n/3$, $\dim (\Ima  [A_k,A_l] + \Ima [A_k,A_m] ) = 3(r-n)$ for any triple of distinct matrices
 $A_k,A_l,A_m$.
 \end{itemize}
 \end{theorem}
 The term "generic" in this theorem has the standard algebro-geometric meaning. Namely, we say that some property
 $\cal P$ holds generically in $K^N$ if there is a nonempty Zariski open subset $O$ of $K^N$ such that $\cal P$ holds
 for all $x \in O$; or equivalently, if 
 the set of points of $K^N$ that do not satisfy $\cal P$ is included in the zero set of some non-identically zero polynomial
 $P \in K[X_1,\ldots,X_N]$. 
 In Theorem~\ref{th:generic}, the points $x \in K^N$ represent tuples of matrices $(R,D_1,\ldots,D_p)$.
 We can therefore take $N=r^2+pr$.
 If each of the properties in a finite list holds generically, then their conjunction also holds generically. 
 We will use this simple fact repeatedly in what follows.
 \begin{remark} \label{rem:construct}
 In order to prove Theorem~\ref{th:generic}, it would be enough to exhibit 
  a {\em single} tuple $\Pi=(R,D_1,\ldots,D_p)$ satisfying
 (i) and (ii). Indeed, the set of matrices $M \in M_{q,s}(K)$ with $\rk M \geq \rho$ is Zariski open in $K^{qs}$ for any $\rho$
 (it is defined by the non-vanishing of some minor of size $\rho$).
 As a result, if (i) and (ii) hold for some tuple $\Pi$ then  $\dim \Ima  [A_k,A_l] \geq 2(r-n)$ and 
 $\dim (\Ima  [A_k,A_l] + \Ima [A_k,A_m] ) \geq 3(r-n)$ hold generically
  (these inequalities hold in a Zariski open set containing $\Pi$). By  Lemma~\ref{lem:2direct} and Proposition~\ref{prop:3direct}, $2(r-n)$ and $3(r-n)$ are the maximum possible dimensions for the respective subspaces.
These inequalities must therefore be equalities.
 \end{remark}
 \begin{remark} \label{rem:cont}
 Our proof of Theorem~\ref{th:generic} is not constructive in the sense of Remark~\ref{rem:construct}, but we use in the same way as in this remark the fact 
 that $\{M; \rk M \geq \rho\}$ is an open set. This fact is often referred to as the  "lower semi-continuity of matrix rank."
 \end{remark}
 \begin{remark} \label{rem:finite}
 If $K$ is a finite field of large enough size (compared to $n$), it is still true that there exist inputs that satisfy the assumption 
 of our uniqueness theorems and admit commuting extensions of size $r$. This is due to the fact that the set of "bad" tuples
 $(R,D_1,\ldots,D_p)$ in the construction of Theorem~\ref{th:generic} is included in the zero set of a polynomial $P {\not \equiv} 0$ of degree
 at most $d(n)$, where $d(n)$ depends only on $n$. If $|K|$is large enough, there will be points in $K^N$ where~$P$
 does not vanish, as shown by e.g. the Schwartz-Zippel lemma (in fact, the proportion of points of $K^N$ where $P$ vanishes will be at most equal to $d(n)/|K|$). The characterization of matrix rank by vanishing minors shows that $d(n)$ grows slowly (polynomially) as a function of $n$.
 \end{remark}
 The remainder of this section is devoted to the proof of Theorem~\ref{th:generic}
 and we complete it in Section~\ref{sec:genproof}.
 Let $V'$ be the matrix made of the first $n$ columns of $R$, and $V$ the matrix made of the last $r-n$ columns.
 Let $U$  be the matrix made of the first $n$ rows of $R^{-1}$, and $U'$ the matrix made of the remaining $r-n$ rows.
 We have for $Z_i$ the block structure:
 \begin{equation} \label{eq:block2}
Z_i=
\begin{pmatrix}
A_i & B_i \\
C_i & D_i
\end{pmatrix} = 
\begin{pmatrix}
UD_iV' & UD_iV \\
U'D_iV' & U'D_iV
\end{pmatrix}
\end{equation}
 As an immediate consequence of the way $U,U',V,V'$ are constructed:
 \begin{lemma} \label{lem:inverse}
 The matrices $U$ and $V'$ are of rank $n$; $U'$ and $V$ are of rank $r-n$. We have $UV'=I_n$, $U'V'=I_{r-n}$,
 $UV=0$, $U'V'=0$. Moreover, $\Ima(V)=\ker(U)$ and $\Ima(V')=\ker(U')$.
 \end{lemma}

\subsection{One commutator} \label{sec:1com}

 If $A$ is a matrix with $r$ rows and $I \subseteq [r]$, we will denote by $A_I$ the submatrix of $A$ made of the rows  indexed with the elements of $I$.
\begin{lemma} \label{lem:polymethod}
Fix two matrices $V \in M_{r,r-n}(K)$ and $W \in M_{r,n}(K)$. 
Let us denote by $M(d_1,\ldots,d_r) \in M_{r}(K)$ the matrix 
having as its first $r-n$ columns the columns of $\diag(d_1,\ldots,d_r)V$, and as its last $n$ columns the columns of $W$.
In block notation, 
$M(d_1,\ldots,d_r)= ( DV\, |\, W)$ where $D=\diag(d_1,\ldots,d_r).$ The two following properties  are equivalent:
\begin{itemize}
\item[(i)] There exists $d_1,\ldots,d_r \in K$ such that $M(d_1,\ldots,d_r)$ is invertible.
\item[(ii)] There is a partition $[r]=I \cup J$ with $|I| = r-n$, $|J| = n$ such that 
$V_I$ and $W_J$ are invertible. 
\end{itemize}
When these two properties hold, $M(d_1,\ldots,d_r)$ is invertible for the specific choice: $d_i=1$ when $i \in I$,
$d_i = 0$ when $i \in J$.
\end{lemma}
\begin{proof}
Consider the polynomial $P(d_1,\ldots,d_r) = \det M(d_1,\ldots,d_r)$.
The first property holds true if and only if  $P$ is not identically 0, i.e., if some monomial
appears in $P$ with a nonzero coefficient. Note that $P$ is a homogeneous multilinear polynomial of degree $r-n$ in the variables $d_1,\ldots,d_r$.
Applying the generalized Laplace expansion along the first $r-n$ columns of $\det M(d_1,\ldots,d_r)$ shows
that the coefficient of the monomial $\prod_{i \in I} d_i$ is equal up to the sign to the product 
$(\det V_I)(\det W_{[r] \setminus I})$. Properties (i) and (ii) are therefore equivalent.

Let us now set $d_i=1$ for $i \in I$, $d_i=0$ otherwise. Suppose without loss of generality that $I = [r-n]$. 
Then $M(d_1,\dots,d_r)$ is of block form
$\begin{pmatrix}
* & * \\
0 & *
\end{pmatrix}$ where the two diagonal blocks are invertible. The block matrix is therefore invertible as well.
\end{proof}

\begin{corollary} \label{cor:full}
Suppose that $r \leq 3n/2$ and that $V \in M_{r,r-n}(K)$ contains at least 3 nonzero row-disjoint minors of size $r-n$.
There exist two diagonal matrices $D_1,D_2$ of size $r$ such that the block matrix $M = ( D_2V \, |\,  D_1V \, |\, V)$ is of full column rank.
\end{corollary}
\begin{proof}
First, note that $M \in M_{r,3(r-n)}(K)$ so $M$ does not have more columns than rows if  $r \leq 3n/2$.
Suppose that the rows of the nonzero minors are indexed by the disjoint subsets $I,J,K \subseteq [r]$. 
Set to 1 all the diagonal entries of $D_1$ indexed by an element of $I$, and set all the others to 0.
Likewise, set to 1 all the diagonal entries of $D_2$ indexed by an element of $K$, and all the others to 0.
The  matrix $M= ( D_2V \, |\,  D_1V \, |\, V)$ is 
of row rank $3(r-n)$. This follows directly from 
the structure of $M$ as a $3 \times 3$ block matrix. Suppose indeed that $K=[r-n]$ and that $I$ is made of the next $r-n$ integers. Then the block structure is upper triangular and the 3 blocks on the diagonal are of rank $r-n$ (this generalizes
the block structure of $M(d_1,\ldots,d_r)$ in the last paragraph of the proof of Lemma~\ref{lem:polymethod}).
Alternatively, $\rk M = 3(r-n)$ follows  from  two successive applications
of Lemma~\ref{lem:polymethod}, beginning with $W=V_{I \cup J \cup K}$.
Since $M$ has $3(r-n)$ columns, it is  of full column rank.
\end{proof}
In the above construction, $D_1$ and $D_2$ are not full rank since each matrix has only $r-n$ nonzero entries.
Nevertheless, once we know that the conclusion of the Corollary holds {\em  for some choice} of $D_1,D_2$,
it must hold also for some pair of invertible matrices. This follows from the fact that $M$ lies in the (open)
  set of matrices of full rank (see also Remarks~\ref{rem:construct} and~\ref{rem:cont}).
   This fact will be useful for the next proposition.
\begin{proposition} \label{prop:directsum}
If $r \leq 3n/2$ one can choose $R \in GL_r(K)$ and two diagonal matrices $D_1,D_2$ of size $r$ 
so that the two top-right blocks $B_1,B_2$ of
$Z_1,Z_2$ in~(\ref{eq:block2})  are of rank $r-n$ and have their images in direct sum.
\end{proposition}
\begin{proof}
We will take $D_1$ and $D_2$ as in Corollary~\ref{cor:full}. For a generic $R \in GL_r(K)$ all minors of all sizes are nonzero,
including those of $V$; Corollary~\ref{cor:full} is therefore applicable. 
Recall from Lemma~\ref{lem:inverse} that $\Ima(V)=\ker(U)$. The 3 subspaces $\Ima(D_1V)$, $\Ima(D_2V)$, $\ker(U)$
are therefore in direct sum. This implies that $\Ima(B_1)=\Ima(UD_1V)$ is in direct sum with $\Ima(B_2)=\Ima(UD_2V)$.
Suppose indeed that $y=UD_1Vx_1=UD_2Vx_2$ for some $x_1,x_2 \in K^{r-n}$. 
Note that $x=D_1Vx_1-D_2Vx_2 \in \ker U$. Since $x \in \Ima(D_1V) \oplus \Ima(D_2V)$ and 
$\Ima(D_1V)$, $\Ima(D_2V)$, $\ker(U)$
are  in direct sum, we conclude that $x=0$. Therefore, $z=D_1Vx_1=D_2Vx_2 \in \Ima(D_1V) \cap \Ima(D_2V)$.
These two subspaces being in direct sum, we must have $z=0$. Hence $y=Uz=0$, and we have shown that 
$\Ima(B_1)$ is in direct sum with $\Ima(B_2)$.

Finally, to ensure that $B_1=UD_1V$ is of rank $r-n$ note that $\ker(B_1)=\ker(D_1V)$ since $\ker U$ and $\Ima D_1V$ are in
direct sum. Hence $\ker(B_1)=\{0\}$ if $D_1$ and $V$ are of full rank; a similar argument applies for $B_2$.
As pointed out before Proposition~\ref{prop:directsum},   it is indeed possible to take $D_1$ and $D_2$ of full rank in 
Corollary~\ref{cor:full}.
Moreover, as pointed out at the beginning of the proof of the proposition, 
$V$ is of full (column) rank for a generic $R \in GL_r(K)$.
\end{proof}
The above proof shows that the conclusion of the proposition holds true for a generic choice of $R$, $D_1$, $D_2$.
Moreover, it applies generically not only to $B_1,B_2$ but also to $C_1^T$, $C_2^T$: for a generic choice of $R$, $D_1$, $D_2$, the images of $C_1^T$ and $C_2^T$ are in direct sum and of dimension $r-n$.
This follows for instance from the fact that the $C_i^T$ are the top right blocks of $Z_i^T=R^TD_iR^{-T}$,
and we can apply Proposition~\ref{prop:directsum} to the $Z_i^T$ instead of the $Z_i$.
This remark will be useful for the next theorem.
\begin{theorem} \label{th:rank2}
If $r \leq 3n/2$ one can choose $R \in GL_r(K)$ and two diagonal matrices $D_1,D_2$ of size $r$ 
so that the two top-left blocks $A_1,A_2$ of $Z_1,Z_2$ in~(\ref{eq:block2}) satisfy $\dim \Ima [A_1,A_2] = 2(r-n)$.
\end{theorem}

\begin{proof}
Recall from~(\ref{eq:commutator}) that $[A_1,A_2]=B_2 C_1 - B_1 C_2$, where $C_1,C_2$ are the bottom left blocks in~(\ref{eq:block2}).
The conclusion of the theorem will hold  if we can choose $R,D_1,D_2$ so that:
\begin{itemize}
\item[(i)] $\Ima(B_1)$, $\Ima(B_2)$ 
are of dimension $r-n$  and in direct sum as in Proposition~\ref{prop:directsum}, 
\item[(ii)] $\Ima(C_1^T)$, $\Ima(C_2^T)$ are of dimension $r-n$, and in direct sum as in~(\ref{eq:direct2}).
\end{itemize}
Suppose indeed that $B_1,B_2,C_1,C_2$ satisfy these two properties. We need to prove that  
$\dim \ker [A_1,A_2] = n-2(r-n)=3n-2r$. 
But $$\ker (B_2 C_1 - B_1 C_2) = \ker(B_2C_1) \cap \ker(B_1C_2)=\ker(C_1) \cap \ker(C_2).$$ 
The first equality follows from the fact that $\Ima B_1$ is in direct sum with $\Ima B_2$, and the second one
from $\ker B_1 = \ker B_2 = \{0\}$. 

Moreover, $\ker(C_1) \cap \ker(C_2) = \ker(C)$ where $C^T$ is defined as
the block matrix $(C_1^T \, |\, C_2^T)$. By (ii), this matrix is of rank $2(r-n)$. Hence $\dim \ker C = n-2(r-n)$ as needed.
We have now proved that (i) and (ii) imply the conclusion of the theorem. Finally,  Proposition~\ref{prop:directsum} shows
 that (i) is satisfied for a generic choice of $R,D_1,D_2$, and we have observed in the remark before Theorem~\ref{th:rank2}
 that (ii) is generically satisfied as well. The conjunction of these two properties will therefore be satisfied for 
 a generic choice of $R,D_1,D_2$.
\end{proof}

\subsection{Two Commutators} \label{sec:2com}

Here we build on the results of Section~\ref{sec:1com} to show that $\Ima [A_1,A_2]+\Ima [A_1,A_3]$ is generically of dimension 
$3(r-n)$. For this, it is  of course necessary to assume that $r \leq 4n/3$. 
Let us begin with another corollary of Lemma~\ref{lem:polymethod}.
\begin{corollary} \label{cor:full2}
Suppose that $r \leq 4n/3$ and that $V \in M_{r,r-n}(K)$ contains at least 4 nonzero row-disjoint minors of size $r-n$.
There exist three diagonal matrices $D_1,D_2,D_3$ of size $r$ such that the block matrix 
$M = ( D_3V \, |\, D_2V \, |\,  D_1V \, |\, V)$ is of full column rank.
\end{corollary}
The proof is omitted since it is essentially identical to that of Corollary~\ref{cor:full}.
\begin{proposition}  \label{prop:directsum2}
If $r \leq 4n/3$ one can choose $R \in GL_r(K)$ and three diagonal matrices $D_1,D_2,D_3$ of size $r$ 
so that the three top right blocks $B_1,B_2,B_3$ of
$Z_1,Z_2,Z_3$ in~(\ref{eq:block2})  are of rank $r-n$ and have their images in direct sum.
\end{proposition}
\begin{proof}
We follow the proof of Proposition~\ref{prop:directsum} and take $D_1,D_2,D_3$ as in Corollary~\ref{cor:full2}. 
As explained after that corollary, these matrices can be taken of rank~$r$ without loss of generality.
This implies $\rk(B_i)=r-n$ exactly as in the proof of Proposition~\ref{prop:directsum}.

Let us now take care of the direct sum property. Since $\Ima(V)=\ker(U)$, the 4  subspaces $\Ima(D_1V)$, $\Ima(D_2V)$, $\Ima(D_3V)$, $\ker(U)$
are  in direct sum. This implies that the 3 subspaces $\Ima(B_1)=\Ima(UD_1V)$, $\Ima(B_2)=\Ima(UD_2V)$,
$\Ima(B_3)=\Ima(UD_3V)$ are in direct sum.
Suppose indeed that $y=UD_1Vx_1=UD_2Vx_2+UD_3Vx_3$ for some $x_1,x_2,x_3 \in K^{r-n}$. 
Note that $x=D_1Vx_1-D_2Vx_2 -D_3Vx_3 \in \ker U$. Since $x \in \Ima(D_1V) \oplus \Ima(D_2V) \oplus \Ima(D_3V)$ and 
$\Ima(D_1V)$, $\Ima(D_2V)$, $\Ima(D_3V)$, $\ker(U)$
are  in direct sum, we conclude that $x=0$. 
Therefore, $z=D_1Vx_1=D_2Vx_2+D_3Vx_3 \in  \Ima(D_1V) \cap (\Ima(D_2V) \oplus \Ima(D_3V))$.
This implies $z=0$ since these 3 subspaces are in direct sum. Hence $y=Uz=0$, and the direct sum property in 
the conclusion of Proposition~\ref{prop:directsum2} is established.
\end{proof}
The remarks following Proposition~\ref{prop:directsum} also apply here: the conclusion of Proposition~\ref{prop:directsum2}
in fact holds true for a generic choice of $R,D_1,D_2,D_3$, and applies not only to $B_1,B_2,B_3$ but also to $C_1^T,C_2^T$, 
$C_3^T$.
\begin{theorem} \label{th:rank3}
If $r \leq 4n/3$ one can choose $R \in GL_r(K)$ and three diagonal matrices $D_1,D_2,D_3$ of size $r$ 
so that the three top-left blocks $A_1,A_2,A_3$ of $Z_1,Z_2,Z_3$ in~(\ref{eq:block2}) satisfy
 $\dim (\Ima [A_1,A_2] + \Ima [A_1,A_3])= 3(r-n)$.
\end{theorem}
\begin{proof}
The conclusion of Theorem~\ref{th:rank3} will hold  if we can choose $R,D_1,D_2,D_3$ so that:
\begin{itemize}
\item[(i)] $\Ima(B_1)$, $\Ima(B_2)$, $\Ima(B_3)$
are of dimension $r-n$  and in direct sum as in Proposition~\ref{prop:directsum2}, 
\item[(ii)] $\Ima(C_1^T)$, $\Ima(C_2^T)$,  $\Ima(C_3^T)$ are of dimension $r-n$ and in direct sum. 
\end{itemize}
Suppose indeed that these two properties are satisfied.
It follows from the proof of Theorem~\ref{th:rank2} that $\rk [A_1,A_2] = \rk [A_1,A_3]=2(r-n)$.
As already pointed out in Lemma~\ref{lem:2direct}, this implies  $\Ima [A_1,A_2] = \Ima(B_1) \oplus \Ima(B_2)$ 
and $\Ima [A_1,A_3] = \Ima(B_1) \oplus \Ima(B_3)$. 
As a result, $\Ima [A_1,A_2] + \Ima [A_1,A_3] =  \Ima(B_1) \oplus \Ima(B_2) \oplus  \Ima(B_3)$. This subspace
is of dimension $3(r-n)$ as needed. Finally, we observe that (i) and (ii) hold generically by Proposition~\ref{prop:directsum2}
and the remark following it.
\end{proof}

\subsection{Proof of the genericity theorem} \label{sec:genproof}

 Theorem~\ref{th:generic} is a fairly straightforward consequence of Theorem~\ref{th:rank2} and Theorem~\ref{th:rank3}.
 The proofs of these theorems show that their conclusions do not only hold for one particular choice of $R,D_1,D_2,D_3$,
 but hold generically. This follows also just from the {\em statements} of these theorems and from lower semi-continuity of
 matrix rank (see Remarks~\ref{rem:construct} and~\ref{rem:cont}).
 Moreover, we can apply these theorems not only to $A_1,A_2,A_3$ but to any pair and triple of matrices in $(A_1,\ldots,A_p)$.
 Each of the properties $\dim \Ima  [A_k,A_l] = 2(r-n)$, $\dim (\Ima  [A_k,A_l] + \Ima [A_k,A_l] ) = 3(r-n)$ therefore holds 
 generically, and their conjunction holds generically as well.

{\small
\section*{Acknowledgments}
The commuting extension in~(\ref{eq:2n}) was communicated by Jeroen Zuiddam.

}

\end{document}